\RequirePackage{fix-cm}
\documentclass[onecolumn,draft]{svjour3}                    
\smartqed  
\usepackage{amsmath}
\usepackage{amsfonts,amssymb,bbm}
\usepackage{enumerate}
\usepackage{pstricks}
\usepackage{dsfont}
\usepackage{scalerel}
\usepackage{mathtools}

\newcommand{\prob}{\mathbb{P}}
\newcommand{\Ex}{\mathbb{E}}
\newcommand{\Rl}{\mathbb{R}}

\newcommand{\gen}{\mathcal{L}}
\newcommand{\process}{\{\eta_t\}_{t\ge 0}}
\newcommand{\jump}{\{\zeta_i\}_{i\ge 0}}

\spnewtheorem{assumption}{Assumption}{\bf}{\it}
\spnewtheorem*{nntheorem}{Theorem}{\bf}{\it}
\spnewtheorem*{nnproposition}{Proposition}{\bf}{\it}

\journalname{Journal of Statistical Physics}

\begin{document}

\title{On the Mean Residence Time in Stochastic Lattice-Gas Models}


\author{Marco Zamparo \and Luca Dall'Asta \and Andrea Gamba}


\institute{Marco Zamparo \at \email{marco.zamparo@polito.it}\\ 
  Dipartimento di Scienza Applicata e Tecnologia, Politecnico di Torino, Torino, Italy\\
  Italian Institute for Genomic Medicine, Torino, Italy
           \and
           Luca Dall'Asta \at \email{luca.dallasta@polito.it}\\
           Dipartimento di Scienza Applicata e Tecnologia, Politecnico di Torino, Torino, Italy\\
           Collegio Carlo Alberto, Universit\`a degli Studi di Torino, Torino, Italy\\
           Italian Institute for Genomic Medicine, Torino, Italy
           \and
           Andrea Gamba \at \email{andrea.gamba@polito.it}\\
           Dipartimento di Scienza Applicata e Tecnologia, Politecnico di Torino, Torino, Italy\\
           Istituto Nazionale di Fisica Nucleare, Sezione di Torino, Torino, Italy\\
           Italian Institute for Genomic Medicine, Torino, Italy
}

\date{}

\maketitle

\begin{abstract}

A heuristic law widely used in fluid dynamics for steady flows states
that the amount of a fluid in a control volume is the product of the
fluid influx and the mean time that the particles of the fluid spend
in the volume, or mean residence time. We rigorously prove that if the
mean residence time is introduced in terms of sample-path averages,
then stochastic lattice-gas models with general injection, diffusion,
and extraction dynamics verify this law. Only mild assumptions are
needed in order to make the particles distinguishable so that their
residence time can be unambiguously defined.  We use our general
result to obtain explicit expressions of the mean residence time for
the Ising model on a ring with Glauber + Kawasaki dynamics and for the
totally asymmetric simple exclusion process with open boundaries.

\keywords{Residence time \and Interacting particle systems \and
  Sample-path averages \and Strong law of large numbers}

\subclass{60F15 \and 60J27 \and 60K35 \and 82C20 \and 82C22}

\end{abstract}

\section{Introduction}

{\it Residence time} is the amount of time that the particles of a
fluid spend in a control volume. Residence time is a ubiquitous
concept involved for instance in the water cycle in hydrology
\cite{Ent}, in the water and wastewater treatment in environmental
engineering \cite{Sincero}, in continuous flow reactions in chemistry
\cite{Nauman}, and in drug kinetics in pharmaceutics \cite{Weiss}.
Beyond fluid dynamics and flow chemistry, the concept of residence
time can be applied to the flow of generic resources from biology up
to economic and social sciences.  Recently, we have used the mean
residence time of proteins on lipid membranes in eukaryotic cells to
evaluate the efficiency of the molecular sorting process
\cite{Zamparo}, whereby specific proteins and lipids are concentrated
and distilled into lipid vesicles.

The mean residence time $\tau$ of a fluid in a fixed control volume is
commonly determined for steady flows through the law $\rho=\phi\tau$
\cite{Ent,Sincero,Nauman}, where $\rho$ is the total amount of fluid
in the volume and $\phi$ is the fluid influx. This law has been
justified on the basis of heuristic arguments, but it has never been
formally expressed and rigorously proven in a microscopic framework
accounting for single fluid particles. A companion principle was
proposed in queuing theory but, unlike the case of fluid dynamics, it
was formulated and demonstrated in a rigorous setting based on
sample-path averages of stochastic queuing processes.  This principle,
which is widely known as {\it Little's law}, states that $l=\lambda w$
\cite{Little}, where $l$ is the mean number of units in the system,
$\lambda$ is their arrival rate, and $w$ is the mean time spent by a
unit in the system.  In this work, we resort to a similar sample-path
formulation to show that the law $\rho=\phi\tau$ for fluids is
rigorously verified in the microscopic framework of stochastic
lattice-gas models with rather general mechanisms of injection,
diffusion, and extraction of particles.  Precisely, $\rho$ and $\phi$
are here the mean number of particles in the system and the influx of
particles in the stationary state.  Mild conditions making particles
distinguishable and trackable must be imposed in order to
unambiguously define their residence time.

Stochastic lattice-gas models are continuous-time Markov processes
describing systems of particles moving in a lattice and interacting
with each other. Since Spitzer's pioneering studies in the late 1960's
on spatially distributed stochastic systems \cite{Griffeath},
stochastic lattice-gas models have become a main subject of research
both in physics, for the deep insight they provide on non-equilibrium
statistical mechanics \cite{Kipnis,Jona}, and in mathematics, for the
new problems they pose in probability theory
\cite{Liggett}. Nevertheless, although issues of existence and
uniqueness have long been settled \cite{Liggett}, proving anything
nontrivial about the properties of such models is surprisingly
difficult. With a few exceptions \cite{Derrida,Schutz}, explicit
calculations are not feasible and one has to be satisfied with Monte
Carlo simulations, qualitative statements based on mean-field
theories, and some explicit bounds.  Complicating the situation is the
fact that the variety of non-equilibrium phenomena one can conceive,
combined with the major role played by the details of the microscopic
dynamics, makes it arduous to define general classes of systems for
which a unified analysis is possible \cite{Jona}. For comparison, the
macroscopic behavior of systems at thermodynamic equilibrium is to a
considerable extent independent of the microscopic details, so that
different systems exhibit qualitatively the same phenomenology at
large scales. In this scenario, our proof of the universal law
$\rho=\phi\tau$ is a breakthrough, in that it provides an exact
connection between distinct dynamical observables in stochastic
lattice-gas models.  It is worth observing here that knowledge of
exact relations is precious in checking the validity of general
polynomial-time approximation schemes, such as mean-field theories
where correlations are neglected.

In order to demonstrate the practical usefulness of the law
$\rho=\phi\tau$, we compute the mean residence time for two well-known
stochastic lattice-gas models: the Ising model on a ring with Glauber
+ Kawasaki dynamics and the totally asymmetric simple exclusion
process with open boundaries. The Ising model is proposed as an
example of a system that is time-reversible at equilibrium, whereas
the totally asymmetric simple exclusion process violates time-reversal
symmetry.  The mean residence time in stochastic lattice-gas models
has been the subject of two recent works, which however ignore, and
therefore do not take advantage of, the exact law $\rho=\phi\tau$.
The first work \cite{Emilio} deals with the mean residence time of
particles undergoing an asymmetric simple exclusion dynamics on a
two-dimensional vertical strip whose top and bottom sides are in
contact with infinite particle reservoirs. In that work, the mean
residence time is approximated numerically and analytically by means
of a mean-field theory and of an analogy with a first-passage-time
problem for a birth-and-death process.  The second work
\cite{Messelink} focuses on the totally asymmetric simple exclusion
process with open boundaries and some of its variants.  In that work,
the {\it on-site} mean residence time, defined as the mean time a
particle spends on a given site before moving on to the next site, is
approximated analytically using a mean-field theory and domain-wall
theory at the coexistence of the low-density phase with the
high-density phase. 
In the case of the standard totally asymmetric simple exclusion
process, Ref.~\cite{Messelink} provides an approximate analytical
expression for the mean residence time, which can be easily computed
as the sum of on-site mean residence times over the entire lattice.
Comparison with Monte Carlo simulations suggests that this approximate
expression is exact in the large system-size limit
\cite{Messelink}. Here we show how the law $\rho=\phi\tau$ applied to
the totally asymmetric simple exclusion process allows to
compute the mean residence time exactly for any system size, obtaining
a result that is perfectly consistent with the findings of
Ref.\ \cite{Messelink} when the system size is sent to infinity.

The paper is organized as follows. In Sect.\ \ref{sec:model} we
introduce the class of stochastic lattice-gas models on which the work
is focused.  Sect.\ \ref{sec:result} is devoted to define the mean
residence time for such stochastic lattice-gas models in terms of
sample-path averages and to state the law $\rho=\phi\tau$ as a limit
theorem for these sample-path averages.  In
Sect.\ \ref{sec:applications} we apply this law to the Ising model on
a ring with Glauber + Kawasaki dynamics and to the totally asymmetric
simple exclusion process with open boundaries. Finally,
Sect.\ \ref{sec:proof} addresses the proof of the law.

\subsection{The Stochastic Lattice-Gas Model}
\label{sec:model}

Let $\Lambda$ be a finite set and let $S$ be the collection of all
functions $\eta:\Lambda\to\{0,1\}$. We will refer to the set $\Lambda$
as the {\it lattice}, to $x\in\Lambda$ as a {\it site} of the lattice,
and to $\eta\in S$ as a {\it microscopic configuration} of the
lattice. Given a site $x\in\Lambda$ and a microscopic configuration
$\eta\in S$, the binary number $\eta(x)$ will be interpreted as the
number of particles of a fluid at $x$ once an exclusion principle is
imposed.  The stochastic lattice-gas model we consider is a
homogeneous continuous-time Markov chain $\process$ defined on some
probability space $(\Omega,\mathcal{F},\prob)$ with {\it state space}
$S$ and right-continuous sample paths.  The microscopic dynamics is
specified by an {\it infinitesimal generator} $\gen$ (see
\cite{Norris}, page 94) that acts on any observable $f:S\to\Rl$
providing its evolution in time in the sense that
\begin{equation*}
\frac{d}{dt}\Ex\big[f(\eta_t)\big]=\Ex\big[\gen f(\eta_t)\big],
\end{equation*}
where $\Ex$ denotes expectation with respect to the probability
measure $\prob$. We write the infinitesimal generator as the
superposition $\gen:=\gen_I+\gen_D+\gen_E$ of a generator $\gen_I$
accounting for injection of particles in the lattice, a generator
$\gen_D$ describing diffusion on the system, and a generator $\gen_E$
governing extraction of particles from the system. Our prescriptions
for the generators $\gen_I$, $\gen_D$, and $\gen_E$ are provided
below. Hereafter, given a state $\eta\in S$ and a subset
$v\subseteq\Lambda$, we denote by $\eta^v$ the microscopic
configuration defined by $\eta^v(x):=1-\eta(x)$ if $x\in v$ and
$\eta^v(x):=\eta(x)$ otherwise.
\begin{assumption}
When the system is in the state $\eta\in S$, then 
\begin{itemize}
\item[\upshape{(I)}] a particle can be injected at a site
  $x\in\Lambda$ with {\upshape injection rate} $i_x(\eta)\ge 0$ if $x$
  is not occupied. Thus, $i_x(\eta)=0$ if $\eta(x)=1$. The action of
  the generator $\gen_I$ on the observable $f$ reads
\begin{equation*}
\gen_If(\eta):=\sum_{x\in\Lambda}i_x(\eta)\big[f(\eta^{\{x\}})-f(\eta)\big];
\end{equation*}
\item[\upshape{(D)}] a particle occupying a site $x$ can diffuse on the system
  through a jump on an empty site $y$ with {\upshape diffusion rate}
  $d_{x,y}(\eta)\ge 0$. Thus, $d_{x,y}(\eta)=0$ if either $\eta(x)=0$
  or $\eta(y)=1$. The action of $\gen_D$ on $f$ is
\begin{equation*}
\gen_D f(\eta):=\sum_{(x,y)\in\Lambda^2}d_{x,y}(\eta)\big[f(\eta^{\{x,y\}})-f(\eta)\big];
\end{equation*}
\item[\upshape{(E)}] all particles in an arbitrary subset
  $v\subseteq\Lambda$ are simultaneously removed from the system with
  {\upshape extraction rate} $e_v(\eta)\ge 0$ if $v$ is completely
  filled. Thus, $e_v(\eta)=0$ if there exists $x\in v$ such that
  $\eta(x)=0$. The action of $\gen_E$ on $f$ reads
\begin{equation}
\nonumber 
\gen_Ef(\eta):=\sum_{v\subseteq\Lambda}e_v(\eta)\big[f(\eta^v)-f(\eta)\big].
\end{equation}
\end{itemize}
\end{assumption}
Some remarks are in order. The injection mechanism described by
$\gen_I$ entails that particles are injected in the system one at a
time. This hypothesis is not necessary to distinguish particles but
largely simplifies the presentation, covering at the same time most of
the interesting physical systems and basically all models found in the
literature. The diffusion mechanism identified by $\gen_D$ accounts
for only one particle jump at a time. This hypothesis is necessary to
distinguish and to track particles, which are unavoidable operations
when one needs to link the particles that leave the system with the
particles that have previously entered in order to define residence
times. Finally, it is worth observing here that the generator $\gen_E$
allows the extraction of any possible subset of the lattice, thus
providing in principle the most general extraction mechanism.

The process $\process$ can be conveniently represented in terms of the
associated jump chain and holding times. It will be important from now
on not to confuse jumps of the system between microscopic
configurations with jumps of the particles on the lattice. Let
$J_0:=0$ and $J_i:=\inf\{t>J_{i-1}:\eta_t\ne \eta_{J_{i-1}}\}$ for
each $i\ge 1$ be the {\it jump times} at which the system moves to a
new state.  We point out that $\lim_{i\uparrow\infty}J_i=\infty$
$\prob$-a.s.\footnote{As usual, we say that a property holds
  $\prob$-{\it almost surely} ($\prob$-a.s.\ for short) if it holds
  for all $\omega\in\Omega_o\in\mathcal{F}$ with $\prob[\Omega_o]=1$.}
because the state space $S$ is finite (see \cite{Norris}, page 90).
Set $\zeta_i:=\eta_{J_i}$ for all $i\ge 0$, so that $\eta_t=\zeta_i$
whenever $t$ satisfies $J_i\le t <J_{i+1}$. The sequence $\jump$
collecting the states that the system progressively visits is called
the {\it jump chain} and results in a homogeneous discrete-time Markov
chain (see \cite{Norris}, page 88). Denote by $q(\eta)$ the real
number defined for each $\eta\in S$ by
\begin{equation}
q(\eta):=\sum_{x\in\Lambda}i_x(\eta)+\sum_{(x,y)\in\Lambda^2}d_{x,y}(\eta)+\sum_{v\subseteq\Lambda}e_v(\eta).
\label{qdef}
\end{equation}
Transition probabilities of the jump chain are given for each $\eta$
and $\eta'$ in $S$ by the formula
$\prob[\zeta_{i+1}=\eta'|\zeta_i=\eta]=\mathds{1}(\eta'=\eta)$ or the
formula
\begin{eqnarray}
\nonumber
\prob\big[\zeta_{i+1}=\eta'\big|\zeta_i=\eta\big]&=&\sum_{x\in\Lambda}\frac{i_x(\eta)}{q(\eta)}\,\mathds{1}(\eta'=\eta^{\{x\}})\\
\nonumber
&+&\sum_{(x,y)\in\Lambda^2}\frac{d_{x,y}(\eta)}{q(\eta)}\,\mathds{1}(\eta'=\eta^{\{x,y\}})\\
&+&\sum_{v\subseteq\Lambda}\frac{e_v(\eta)}{q(\eta)}\,\mathds{1}(\eta'=\eta^v)
\label{probjump}
\end{eqnarray}
depending on whether $q(\eta)=0$ or $q(\eta)>0$ (see \cite{Norris},
page 87). The time $H_i$ that the process $\process$ spends in the
state $\zeta_i$ is $H_i:=J_{i+1}-J_i$ and is called {\it holding
  time}.  For each $i\ge 0$, conditional on $\zeta_0,\ldots,\zeta_i$,
the holding times $H_0,\ldots,H_i$ are independent exponential random
variables of parameters $q(\zeta_0),\ldots,q(\zeta_i)$ respectively
(see \cite{Norris}, page 88).

The process $\process$ is said to be {\it irreducible} if for each
microscopic configurations $\eta$ and $\eta'$ there exists an integer
$i\ge 0$ such that $\prob\big[\zeta_i=\eta'|\zeta_0=\eta]>0$.  Most of
the physical phenomena that can be described in terms of stochastic
lattice-gas models originate processes that do not become trapped in
proper subsets of the state space, thus resulting irreducible
\cite{Jona}.  Irreducibility is assumed here.
\begin{assumption}
The process $\process$ is irreducible.
\end{assumption}
Irreducibility combined with the fact that the state space $S$ is
finite due to the finiteness of $\Lambda$ has a number of
consequences. First of all, no state is {\it absorbing}, meaning that
$q(\eta)>0$ for all $\eta\in S$. This gives in particular that the
jump chain $\jump$ cannot stay at rest, satisfying for each $i\ge 1$
one of the following alternatives $\prob$-a.s.:
\begin{itemize}
\item[(I)] there exists a site $x\in\Lambda$ such that
  $\zeta_{i-1}(x)=0$ and $\zeta_i=\zeta_{i-1}^{\{x\}}$;
\item[(D)] there exist $x$ and $y$ in $\Lambda$ such that $\zeta_{i-1}(x)=1$,
  $\zeta_{i-1}(y)=0$, and $\zeta_i=\zeta_{i-1}^{\{x,y\}}$; \\[-0.85em]
\item[(E)] there exists a cluster $v\subseteq\Lambda$ such that
  $\zeta_{i-1}(x)=1$ for all $x\in v$ and
  $\zeta_i=\zeta_{i-1}^v$.
\end{itemize}
Secondly, the jump chain is {\it recurrent} (see \cite{Norris}, page
27), so that for every $\eta\in S$ there exist
$\prob$-a.s.\ infinitely many $i$ with the property that
$\zeta_i=\eta$. Third, there exists a unique invariant distribution
$\pi$ (see \cite{Norris}, page 118).  We recall that a distribution
$\pi$ on $S$ is {\it invariant} if $\sum_{\eta\in S}\gen
f(\eta)\pi(\eta)=0$ for all observables $f:S\to\Rl$. Last, strong laws
of large numbers hold for functionals of the process $\process$ (see
\cite{Norris}, page 126) and the jump chain $\jump$ (see
\cite{Serfozo}, page 267).

\subsection{The Mean Residence Time Law}
\label{sec:result}

Let $\theta_i(x,y)$ be the binary random variable defined for each
integer $i\ge 1$ and sites $x$ and $y$ in $\Lambda$ by
\begin{equation*}
  \theta_i(x,y):=\begin{cases}
  \zeta_{i-1}(x)\zeta_i(x) & \mbox{if }x=y;\\
  \zeta_{i-1}(x)[1-\zeta_{i-1}(y)][1-\zeta_i(x)]\zeta_i(y) & \mbox{if }x\ne y,
\end{cases}
\end{equation*}
where $\jump$ is the jump chain.  Considering separately the
alternatives (I), (D), and (E) for the $i$th configuration jump it is
not difficult to verify that $\theta_i(x,x)=1$ if and only if $x$
hosts a particle that stays at rest during this change of
configuration. Similarly, $\theta_i(x,y)=1$ with $x\ne y$ if and only
if there is a particle at $x$ that moves to $y$ in the $i$th
configuration jump. Tracking particles on the lattice is now possible.
Given a couple of integers $j\ge i\ge 1$ and sites
$x_{i-1},\ldots,x_j$ in $\Lambda$ not necessarily distinct, we have
that $x_{i-1}$ hosts a particle that moves progressively to $x_k$ in
the $k$th configuration jump with $k$ running from $i$ to $j$ if and
only if $\prod_{k=i}^j\theta_k(x_{k-1},x_k)=1$. It follows in
particular that a particle located at $x_{i-1}$ before the $i$th
configuration jump is still in the system after the $j$th change of
configuration if and only if there exist $x_i,\ldots,x_j$ in $\Lambda$
such that $\prod_{k=i}^j\theta_k(x_{k-1},x_k)=1$. This condition is
tantamount to $\Theta_{i,j}(x_{i-1})=1$, where $\Theta_{i,j}(x_{i-1})$
is the binary random variable defined for each $j\ge i\ge 1$ and
$x_{i-1}\in\Lambda$ by
\begin{equation}
\Theta_{i,j}(x_{i-1}):=\sum_{x_i\in\Lambda}\cdots\sum_{x_j\in\Lambda}\prod_{k=i}^j\theta_k(x_{k-1},x_k).
\label{Thetadef}
\end{equation}

Let $|\eta|:=\sum_{x\in\Lambda}\eta(x)$ denote the number of particles
in the system under the state $\eta\in S$.  Given an integer $i\ge 1$,
a particle enters the system in the $i$th configuration jump if and
only if the condition $|\zeta_i|>|\zeta_{i-1}|$ that excludes
alternatives (D) and (E) is fulfilled.  If $|\zeta_i|>|\zeta_{i-1}|$,
then the particle is injected at that unique site $x$ such that
$\zeta_{i-1}(x)=0$ and $\zeta_i(x)=1$. This way, for each $j\ge i$ we
can state that a particle enters the system in the $i$th configuration
jump and is still in the system after the $j$th change of
configuration if and only if $U_{i,j}=1$, where $U_{i,j}$ is the
binary random variable defined by
\begin{equation}
U_{i,j}:=
\begin{cases}
\mathds{1}(|\zeta_i|>|\zeta_{i-1}|) & \mbox{if }j=i; \\
\mathds{1}(|\zeta_i|>|\zeta_{i-1}|)\sum_{x\in\Lambda}[1-\zeta_{i-1}(x)]\Theta_{i+1,j}(x) & \mbox{if }j>i.
\end{cases}
\label{Udef}
\end{equation}
For any $j>i\ge 1$, we have that $U_{i,j-1}\ge U_{i,j}$ and that a
particle enters the system in the $i$th configuration jump and leaves
it exactly in the $j$th change of configuration if and only if
$U_{i,j-1}-U_{i,j}=1$.  The random variables (\ref{Thetadef}) and
(\ref{Udef}) satisfy $\lim_{j\uparrow\infty}\Theta_{i,j}(x)=0$
$\prob$-a.s.\ and $\lim_{j\uparrow\infty}U_{i,j}=0$ $\prob$-a.s. for
any $i\ge 1$ and $x\in\Lambda$, meaning that every particle eventually
leaves the system. The simplest way to prove these limits is to
observe that recurrence of the jump chain implies that there exist
$\prob$-a.s.\ infinitely many $k$ such that $\zeta_k(x)=0$ for all
$x$. For such $k$ it holds that $\theta_k(x,y)=0$ for all $x$ and $y$
in $\Lambda$.

We are now able to define the mean residence time in terms of
sample-path averages. A particle that enters the system in the $i$th
configuration jump and leaves it exactly in the $j$th change of
configuration spends in the lattice the time $J_j-J_i$. Thus, denoting
by $N_t:=\sup\{i\ge 0:J_i\le t\}$ the number of configuration jumps up
to a certain time $t\ge 0$, we introduce the mean residence time $T_t$
of the particles that have been injected by the time $t$ as
\begin{equation}
  T_t:=\frac{\sum_{i=1}^{N_t}\sum_{j=i+1}^\infty (J_j-J_i)(U_{i,j-1}-U_{i,j})}{\sum_{i=1}^{N_t}\sum_{j=i+1}^\infty (U_{i,j-1}-U_{i,j})}.
\label{MRT}
\end{equation}
Hereafter we assume that a sum with upper limit smaller than the lower
one is equal to zero and that $0/0:=0$. The following theorem stating
the law $\rho=\phi\tau$ for stochastic lattice-gas models is our main
result.
\begin{nntheorem}
  \label{main1}
  Let $\pi$ be the invariant distribution of $\gen$ and set
  $\rho:=\sum_{\eta\in S}|\eta|\pi(\eta)$ and $\phi:=\sum_{\eta\in
    S}\sum_{x\in\Lambda}i_x(\eta)\pi(\eta)$. Then, the limit
  $\lim_{t\uparrow\infty} T_t=:\tau$ exists $\prob$-a.s.\ and
  satisfies $\rho=\phi\tau$.
\end{nntheorem}

The real number $\rho$ is the mean number of particles in the system
with respect to $\pi$.  The real number $\phi$ is the rate at which
particles enter the system measured as follows. The number of
particles that are injected in the lattice by the time $t\ge 0$ is
$\sum_{i=1}^{N_t}U_{i,i}=\sum_{i=1}^{N_t}\mathds{1}(|\zeta_i|>|\zeta_{i-1}|)$.
Thus, appealing to the strong law of large numbers for functionals of
the jump chain (see \cite{Serfozo}, page 267) first and to the
explicit expression (\ref{probjump}) of its transition probabilities
later we get
\begin{eqnarray}
\nonumber
\lim_{t\uparrow\infty}\,\frac{1}{t}\sum_{i=1}^{N_t}U_{i,i}
&=&\sum_{\eta\in S}\sum_{\eta'\in S}\mathds{1}(|\eta'|>|\eta|)\,\prob\big[\zeta_1=\eta'\bigl|\zeta_0=\eta\big]\,
q(\eta)\,\pi(\eta)~~~~~~~\prob\mbox{-a.s.}\\
&=&\sum_{\eta\in S}\sum_{x\in\Lambda}i_x(\eta)\pi(\eta)=\phi.
\label{phi_rate}
\end{eqnarray}

\subsection{Applications}
\label{sec:applications}

In this section we make use of the law $\rho=\phi\tau$ to compute
explicitly the mean residence time for two well-known stochastic
lattice-gas models. The first model is the Ising model on a ring with
Glauber + Kawasaki dynamics, which is proposed as an example of a
system that is time-reversible at equilibrium. The second model is the
totally asymmetric simple exclusion process with open boundaries,
which violates time-reversal symmetry. We recall that the process
$\process$ is said to be {\it time-reversible} if $\{\eta_t\}_{0\le
  t\le T}$ and $\{\eta_{T-t}\}_{0\le t\le T}$ share the same
finite-dimensional marginal distributions for any number $T>0$.  The
irreducible homogeneous continuous-time Markov chain $\process$ with
invariant distribution $\pi$ is time-reversible if and only if
$\prob[\eta_0=\eta]=\pi(\eta)$ for all $\eta\in S$, so that $\process$
is stationary, and the infinitesimal generator $\gen$ satisfies
detailed balance with respect to $\pi$ (see \cite{Norris}, page
125). The generator $\gen$ is said to satisfy {\it detailed balance}
with respect to a probability distribution $\lambda$ on $S$ if
$\sum_{\eta\in S}(f\gen g-g\gen f)(\eta)\lambda(\eta)=0$ for every two
observables $f:S\to\Rl$ and $g:S\to\Rl$. The distribution $\lambda$ is
invariant if $\gen$ satisfies detailed balance with respect to
$\lambda$ (see \cite{Norris}, page 125).

Conditions on the rates providing reversibility can be easily obtained
as follows.  Assume that $\gen$ satisfies detailed balance with
respect to $\pi$ and observe that irreducibility and finiteness of the
state space entail $\pi(\eta)>0$ for all $\eta\in S$ (see
\cite{Norris}, page 118). Given a site $\bar{x}$ and a state
$\bar{\eta}$ such that $\bar{\eta}(\bar{x})=1$, the condition
$\sum_{\eta\in S}(f\gen g-g\gen f)(\eta)\lambda(\eta)=0$ results in
$e_{\{\bar{x}\}}(\bar{\eta})\pi(\bar{\eta})-
i_{\bar{x}}(\bar{\eta}^{\{\bar{x}\}})\pi(\bar{\eta}^{\{\bar{x}\}})=0$
when $f(\eta):=\mathds{1}(\eta=\bar{\eta})$ and
$g(\eta):=\mathds{1}(\eta=\bar{\eta}^{\{\bar{x}\}})$ for all $\eta$.
The arbitrariness of $\bar{x}$ and $\bar{\eta}$, combined with the
fact that $e_{\{\bar{x}\}}(\eta)=i_{\bar{x}}(\eta^{\{\bar{x}\}})=0$ by
construction if $\eta(\bar{x})=0$, yields $e_{\{x\}}(\eta)\pi(\eta)=
i_{x}(\eta^{\{x\}})\pi(\eta^{\{x\}})$ for every $x\in\Lambda$ and
$\eta\in S$. Furthermore, if $\bar{v}$ is a set of at least two sites
and $\bar{\eta}$ is a microscopic configuration such that
$\bar{\eta}(x)=1$ for each $x\in\bar{v}$, then the choice
$f(\eta):=\mathds{1}(\eta=\bar{\eta})$ and
$g(\eta):=\mathds{1}(\eta=\bar{\eta}^{\bar{v}})$ for any $\eta$ shows
that $0=\sum_{\eta\in S}(f\gen g-g\gen
f)(\eta)\pi(\eta)=e_{\bar{v}}(\bar{\eta})\pi(\bar{\eta})$.  The
arbitrariness of $\bar{v}$ and $\bar{\eta}$ and the fact that
$e_{\bar{v}}(\eta)=0$ by construction if $\eta(x)=0$ for some
$x\in\bar{v}$ imply that $e_v(\eta)=0$ for all $v\subseteq\Lambda$
containing more than one site and all $\eta\in S$. In conclusion, we
find that for each $v\subseteq\Lambda$ and $\eta\in S$
\begin{equation}
  e_v(\eta)=
  \begin{cases}
    i_x(\eta^{\{x\}})\frac{\pi(\eta^{\{x\}})}{\pi(\eta)} & \mbox{if $v=\{x\}$ for some $x\in\Lambda$}; \\
    0 & \mbox{otherwise}.
  \end{cases}
\label{Gdyn}
\end{equation}
Such extraction rates make the dynamics generated by $\gen_I+\gen_E$ a
{\it Glauber dynamics} \cite{Presutti}, whereby only the update of one
site at a time is involved and the detailed balance condition is
fulfilled. Given now two distinct sites $\bar{x}$ and $\bar{y}$ and a
microscopic configuration $\bar{\eta}$ such that
$\bar{\eta}(\bar{x})=1$ and $\bar{\eta}(\bar{y})=0$, the condition
$\sum_{\eta\in S}(f\gen g-g\gen f)(\eta)\pi(\eta)=0$ with
$f(\eta):=\mathds{1}(\eta=\bar{\eta})$ and
$g(\eta):=\mathds{1}(\eta=\bar{\eta}^{\{\bar{x},\bar{y}\}})$ for all
$\eta$ becomes $d_{\bar{x},\bar{y}}(\bar{\eta})\pi(\bar{\eta})-
d_{\bar{y},\bar{x}}(\bar{\eta}^{\{\bar{x},\bar{y}\}})\pi(\bar{\eta}^{\{\bar{x},\bar{y}\}})=0$.
The arbitrariness of $\bar{x}$, $\bar{y}$, and $\bar{\eta}$, combined
with the fact that
$d_{\bar{x},\bar{y}}(\eta)=d_{\bar{y},\bar{x}}(\eta^{\{\bar{x},\bar{y}\}})=0$
by construction if either $\eta(\bar{x})=0$ or $\eta(\bar{y})=1$, leads
to the relationship
\begin{equation}
d_{x,y}(\eta)\pi(\eta)=d_{y,x}(\eta^{\{x,y\}})\pi(\eta^{\{x,y\}})
\label{Kdyn}
\end{equation}
to be satisfied for all $x$ and $y$ in $\Lambda$ and $\eta\in S$.  The
dynamics generated by $\gen_D$ is called a {\it Kawasaki dynamics}
\cite{Presutti} if the set $\Lambda$ is endowed with a graph structure
and if the diffusion rates satisfy both (\ref{Kdyn}) and the property
that $d_{x,y}(\eta)=0$ for all $\eta$ whenever $x$ and $y$ are not
nearest-neighbor sites.

Conditions (\ref{Gdyn}) and (\ref{Kdyn}) are necessary conditions for
the generator $\gen$ to satisfy detailed balance with respect to the
distribution $\pi$. Simple algebra shows that they also are sufficient
conditions to give $\sum_{\eta\in S}(f\gen g-g\gen
f)(\eta)\lambda(\eta)=0$ for all $f:S\to\Rl$ and $g:S\to\Rl$.  Thus,
$\gen$ satisfies detailed balance with respect to $\pi$ if and only if
(\ref{Gdyn}) and (\ref{Kdyn}) hold.

\subsubsection{The Ising Model on a Ring with Glauber + Kawasaki Dynamics}

Let $\Lambda:=\mathbb{Z}/L\mathbb{Z}$ be the one-dimensional discrete
torus of size $L\ge 2$ and let the function $\mathcal{H}:S\to\Rl$ be
the {\it Ising Hamiltonian} defined for each $\eta\in S$ by
\begin{equation}
\mathcal{H}(\eta):=V\sum_{x\in\Lambda}\eta(x)\eta(x+1)-\mu\sum_{x\in\Lambda}\eta(x),
\label{HIsing}
\end{equation}
where $V\in\Rl$ is the {\it interaction parameter} and $\mu\in\Rl$ is
the {\it chemical potential}. The {\it Gibbs state} associated to
$\mathcal{H}$ is the distribution
$\pi_\mathrm{G}:=(1/Z)\exp(-\mathcal{H})$, $Z$ being the {\it
  partition function}. In this section we consider a stochastic
lattice-gas model whose generator $\gen$ satisfies detailed balance
with respect to the Gibbs state $\pi_\mathrm{G}$, so that extraction
and diffusion rates fulfill the conditions (\ref{Gdyn}) and
(\ref{Kdyn}) respectively with $\pi=\pi_\mathrm{G}$. For simplicity,
we focus here on the local and translationally invariant injection
rates defined for all $x\in\Lambda$ and $\eta\in S$ by the formula
\begin{equation*}
i_x(\eta):=\big[1-\eta(x)\big]\alpha_{\eta(x-1),\eta(x+1)},
\end{equation*}
where the parameters $\alpha_{0,0}$, $\alpha_{1,0}$, $\alpha_{0,1}$,
and $\alpha_{1,1}$ are assumed to be strictly positive.  Non-vanishing
extraction rates inherit the same local and translationally invariant
structure, since combining (\ref{Gdyn}) with (\ref{HIsing}) we get
$e_{\{x\}}(\eta)=\eta(x)\beta_{\eta(x-1),\eta(x+1)}$ for any $x$ and
$\eta$ with the strictly positive parameters
$\beta_{0,0}:=\alpha_{0,0}e^{-\mu}$,
$\beta_{1,0}:=\alpha_{1,0}e^{V-\mu}$,
$\beta_{0,1}:=\alpha_{0,1}e^{V-\mu}$, and
$\beta_{1,1}:=\alpha_{1,1}e^{2V-\mu}$. Since $i_x(\eta)>0$ if
$\eta(x)=0$ and $e_{\{x\}}(\eta)>0$ if $\eta(x)=1$ for every $x$ and
$\eta$, the process $\process$ turns out to be irreducible
irrespective of the features of the generator $\gen_D$.  Although the
mean residence time does not depend on the details of diffusion rates
as long as condition (\ref{Kdyn}) holds, to fix the ideas we consider
here the Kawasaki dynamics where particles can only jump to
nearest-neighbor sites. We will refer to this model as the Ising model
on a ring with Glauber + Kawasaki dynamics.

The mean residence time $\tau$ can be computed explicitly as follows.
The translational symmetry of the invariant distribution and of the
injection rates yields $\rho=L\sum_{\eta\in S}\eta(1)\pi(\eta)$ and
$\phi=L\sum_{\eta\in S}i_1(\eta)\pi(\eta)$. Consequently, we have
\begin{equation}
  \tau=\frac{\rho}{\phi}=\frac{\sum_{\eta\in S}\eta(1)\exp[-\mathcal{H}(\eta)]}{\sum_{\eta\in S}i_1(\eta)\exp[-\mathcal{H}(\eta)]}.
\label{startIsing}
\end{equation}
The sums over $\eta$ that appear in (\ref{startIsing}) can be carried
out by means of the transfer matrix method. Let
$\mathcal{T}\in\Rl^{2\times 2}$ be the symmetric matrix with entries
$\mathcal{T}_{0,0}:=1$,
$\mathcal{T}_{1,0}=\mathcal{T}_{0,1}:=e^{\mu/2}$, and
$\mathcal{T}_{1,1}:=e^{\mu-V}$. The matrix $\mathcal{T}$ allows us
to recast the weight $\exp[-\mathcal{H}(\eta)]$ as
$\prod_{x\in\Lambda} \mathcal{T}_{\eta(x),\eta(x+1)}$ for each
$\eta\in S$. This way, we get
\begin{equation}
\sum_{\eta\in S}\eta(1)\exp\big[-\mathcal{H}(\eta)\big]=\sum_{\eta\in S}\eta(1)\prod_{x\in\Lambda}\mathcal{T}_{\eta(x),\eta(x+1)}=(\mathcal{T}^L)_{1,1}
\label{numIsing}
\end{equation}
and
\begin{eqnarray}
\nonumber
  \sum_{\eta\in S}i_1(\eta)\exp\big[-\mathcal{H}(\eta)\big]&=&
  \sum_{\eta\in S}\big[1-\eta(1)\big]\alpha_{\eta(0),\eta(2)}\prod_{x\in\Lambda}\mathcal{T}_{\eta(x),\eta(x+1)}\\
\nonumber
&=&\alpha_{0,0}(\mathcal{T}^{L-2})_{0,0}+\alpha_{1,0}e^{\mu/2}(\mathcal{T}^{L-2})_{1,0}+\\
&&\alpha_{0,1}e^{\mu/2}(\mathcal{T}^{L-2})_{0,1}+\alpha_{1,1}e^\mu(\mathcal{T}^{L-2})_{1,1}.~~~~~~~
\label{denIsing}
\end{eqnarray}
We now use the fact that $\mathcal{T}$ is symmetric to write down for
any $n\ge 0$ the spectral decomposition
$\mathcal{T}^n=t_+^n\mathcal{P}_+^{\phantom{n}}+t_-^n\mathcal{P}_-^{\phantom{n}}$,
where $t_-<t_+$ are the eigenvalues of $\mathcal{T}$ and
$\mathcal{P}_-$ and $\mathcal{P}_+$ are the orthogonal projections
onto the corresponding eigenspaces. The eigenvalues and the
projections are given by the formulas
\begin{equation*}
t_\pm=\frac{1+e^{\mu-V}\pm\sqrt{[1-e^{\mu-V}]^2+4e^\mu}}{2}
\end{equation*}
and $\mathcal{P}_\pm=(\mathcal{T}-t_\mp\mathcal{I})/(t_\pm-t_\mp)$,
$\mathcal{I}$ being the identity matrix.  Thus, combining
(\ref{startIsing}) with (\ref{numIsing}) and (\ref{denIsing}) first
and making use of this spectral decomposition later we reach the
result
\begin{equation}
\tau=\frac{r_+^{\phantom{L}}t_+^L+r_-^{\phantom{L}}t_-^L}{a_+^{\phantom{L}}t_+^{L-2}+a_-^{\phantom{L}}t_-^{L-2}},
\label{tauIsing}
\end{equation}
where
\begin{equation*}
r_\pm:=(\mathcal{P}_\pm)_{1,1}=\frac{e^{\mu-V}-t_{\mp}}{t_\pm-t_\mp}=\frac{(t_\pm-1)^2}{(t_\pm-1)^2+e^\mu}
\end{equation*}
and
\begin{eqnarray}
  \nonumber
  a_\pm&:=&\alpha_{0,0}(\mathcal{P}_\pm)_{0,0}+\alpha_{1,0}e^{\mu/2}(\mathcal{P}_\pm)_{1,0}
  +\alpha_{0,1}e^{\mu/2}(\mathcal{P}_\pm)_{0,1}+\alpha_{1,1}e^\mu(\mathcal{P}_\pm)_{1,1}\\
  \nonumber
  &=&\frac{\alpha_{0,0}(1-t_\mp)+(\alpha_{1,0}+\alpha_{0,1})e^\mu+\alpha_{1,1}[e^{\mu-V}-t_\mp]}{t_\pm-t_\mp}\\
  \nonumber
  &=&\frac{\alpha_{0,0}+(\alpha_{1,0}+\alpha_{0,1})(t_\pm-1)+\alpha_{1,1}(t_\pm-1)^2}{1+e^{-\mu}(t_\pm-1)^2}.
\end{eqnarray}
The explicit expression of the mean residence time $\tau$ for the
Ising model on a ring with Glauber + Kawasaki dynamics is thus
provided by (\ref{tauIsing}). The time $\tau$ is bounded with respect
to the system size $L$ because particles can leave the system at each
site. A different situation is observed in the totally asymmetric
simple exclusion process, where particles have to travel a macroscopic
distance before being allowed to leave the system.

\subsubsection{The Totally Asymmetric Simple Exclusion Process}

Let $\Lambda$ be the set $\{1,\ldots,L\}$ for some integer $L\ge 2$.
The totally asymmetric simple exclusion process with open boundaries
is the irreducible stochastic lattice-gas model associated with the
lattice $\Lambda$ and the following rates, where $\alpha>0$ and
$\beta>0$ are model parameters: $i_1(\eta):=\alpha[1-\eta(1)]$ and
$i_x(\eta):=0$ if $x>1$ as far as injection rates are concerned,
$d_{x,x+1}(\eta):=\eta(x)[1-\eta(x+1)]$ if $x<L$ and
$d_{x,y}(\eta):=0$ if $y\ne x+1$ for diffusion rates,
$e_{\{L\}}(\eta):=\beta\eta(L)$ and $e_v(\eta):=0$ if $v\ne\{L\}$ for
extraction rates. Thus, particles enter the lattice $\Lambda$ at the
left boundary with rate $\alpha$, can move rightwards, and leave the
system at the right boundary with rate $\beta$.  The invariant
distribution $\pi$ is known \cite{Derrida} and an explicit expression
for the probability with respect to $\pi$ that a certain site is
occupied can be obtained \cite{Derrida}. We need this expression in
order to compute the mean residence time. Let $Z_x$ be the real number
defined for each integer $x\ge 0$ by
\begin{equation*}
  Z_x:=\begin{cases}
  1 & \mbox{if }x=0;\\
  \sum_{k=1}^x B_{x,k}\sum_{l=0}^k\frac{1}{\alpha^l}\frac{1}{\beta^{k-l}} & \mbox{if }x\ge 1,
\end{cases}
  \end{equation*}
where $B_{x,k}$ is the combinatorial coefficient given for all $x\ge
1$ and $k\ge 1$ by the formula
\begin{equation*}
B_{x,k}:=\frac{k(2x-k-1)!}{x!(x-k)!}.
\end{equation*}
The probability $\sum_{\eta\in S}\eta(x)\pi(\eta)$ with respect to
$\pi$ that a generic site $x\in\Lambda$ is occupied is \cite{Derrida}
\begin{equation}
  \sum_{\eta\in S}\eta(x)\pi(\eta)=\frac{1}{Z_L}
  \begin{cases}
  \sum_{k=1}^{L-x}\big[Z_{L-k}B_{k,1}+\frac{Z_{x-1}B_{L-x,k}}{\beta^{k+1}}\big] & \mbox{if }x<L;\\
  \frac{Z_{L-1}}{\beta} & \mbox{if }x=L.
  \end{cases}
  \label{tasep}
\end{equation}

The mean residence time of the totally asymmetric simple exclusion
process with open boundaries can be immediately determined by
combining the $\rho=\phi\tau$ law with (\ref{tasep}). To get at a more
compact expression, we notice that the influx $\phi$ defined as
$\phi:=\alpha\sum_{\eta\in S}[1-\eta(1)]\pi(\eta)$ equals
$\beta\sum_{\eta\in S}\eta(L)\pi(\eta)$. Indeed, from the definition
of $\pi$ we have that $0=\sum_{\eta\in S}\gen
f(\eta)\pi(\eta)=\phi-\beta\sum_{\eta\in S}\eta(L)\pi(\eta)$ if
$f(\eta):=|\eta|$ for each $\eta\in S$. This way, we can write
\begin{eqnarray}
  \nonumber
  \tau&=&\frac{\rho}{\phi}=\frac{\sum_{x=1}^L\sum_{\eta\in S}\eta(x)\pi(\eta)}{\beta\sum_{\eta\in S}\eta(L)\pi(\eta)}\\
  &=&\frac{1}{\beta}+\frac{1}{Z_{L-1}}\sum_{x=1}^{L-1}\Bigg[(L-x)Z_{L-x}B_{x,1}+\sum_{k=1}^{L-x}\frac{Z_{x-1}B_{L-x,k}}{\beta^{k+1}}\Bigg].
\label{tauTASEP}
\end{eqnarray}
This formula provides the exact mean residence time $\tau$ for any
system size $L\ge 2$. Even though (\ref{tauTASEP}) is slightly
cumbersome to deal with, asymptotic analysis shows that $\tau$ is
proportional to $L$ in the large $L$ limit with the simple coefficient
of proportionality $r$ given by
\begin{equation*}
  r:=
  \begin{cases}
    2 & \mbox{if $\alpha\ge 1/2$ and $\beta\ge 1/2$};\\
    \frac{1}{2\alpha(1-\alpha)} & \mbox{if $\alpha=\beta<1/2$};\\
    \frac{1}{1-\alpha} & \mbox{if $\alpha<1/2$ and $\alpha<\beta$};\\
    \frac{1}{\beta} & \mbox{if $\beta<1/2$ and $\beta<\alpha$}.
  \end{cases}
\end{equation*}
Indeed, the following proposition holds, confirming that $\tau$ is
proportional to the distance $L$ that particles have to travel before
leaving the system.
\begin{nnproposition}
For each $\alpha>0$ and $\beta>0$ there exists a positive constant
$c<\infty$ independent of $L$ such that
\begin{equation*}
\bigg|\frac{\tau}{L}-r\bigg|\le \frac{c}{\sqrt{L}}.
\end{equation*}
\end{nnproposition}
The proof of this proposition goes through the asymptotic analysis of
the number $Z_x$ in the large $x$ limit, which can be performed by
means of Laplace's method for sums as in Ref.\ \cite{Derrida}. We omit
the details because they are easily imaginable and not very
informative. We point out that the coefficient $r$ has been previously
determined in Ref.\ \cite{Messelink}, where the mean time that a
particle spends on a given site before moving on to the next site has
been investigated by means of mean-field theory. In particular, it has
been shown there by comparison with Monte Carlo simulations that a
mean-field theory neglecting time correlations in the local density of
particles provides the exact value of $r$ for all $\alpha>0$ and
$\beta>0$, except for the case $\alpha=\beta<1/2$ where it fails. The
coefficient of $r$ in the case $\alpha=\beta<1/2$, corresponding to
coexistence between a low-density phase and a high-density phase, has
been found in Ref.\ \cite{Messelink} by combining mean-field
estimations with domain-wall theory.

\section{Proof of the Mean Residence Time Law}
\label{sec:proof}

In this section we prove that $\lim_{t\uparrow\infty} T_t=\rho/\phi$
$\prob$-a.s., thus demonstrating the $\rho=\phi\tau$ law.  We first
observe that the denominator of (\ref{MRT}) divided by $t$ tends to
$\phi$ in the large $t$ limit since the number of particles that enter
the system equals the number of particles that enter the system and
eventually leave it.  Formally,
$\lim_{t\uparrow\infty}(1/t)\sum_{i=1}^{N_t}\sum_{j=i+1}^\infty
(U_{i,j-1}-U_{i,j})=\phi$ $\prob$-a.s.\ follows from (\ref{phi_rate})
since $\sum_{j=i+1}^\infty(U_{i,j-1}-U_{i,j})=U_{i,i}$
$\prob$-a.s.\ for each $i\ge 1$ due to the fact that
$\lim_{j\uparrow\infty}U_{i,j}=0$ $\prob$-a.s.. This way, in order to
prove that $\lim_{t\uparrow\infty} T_t=\rho/\phi$ $\prob$-a.s.\ it
suffices to show that
$\lim_{t\uparrow\infty}(1/t)\sum_{i=1}^{N_t}\sum_{j=i+1}^\infty
(J_j-J_i)(U_{i,j-1}-U_{i,j})=\rho$ $\prob$-a.s..  The latter limit is
verified if
\begin{equation}
\lim_{t\uparrow\infty}\,\frac{1}{t}\sum_{i=1}^{N_t}\sum_{j=i}^\infty H_jU_{i,j}=\rho~~~~~~~\prob\mbox{-a.s.}.
\label{secondo}
\end{equation}
Indeed, we have
$\sum_{j=i+1}^\infty(J_j-J_i)(U_{i,j-1}-U_{i,j})=\sum_{j=i}^\infty H_j
U_{i,j}$ $\prob$-a.s.\ for every $i\ge 1$, thanks to the identity
$J_j-J_i=\sum_{k=i}^{j-1} H_k$ and the fact that $U_{i,j}=0$ for all
sufficiently large $j$ $\prob$-a.s., since on the one hand
$\lim_{j\uparrow\infty}U_{i,j}=0$ $\prob$-a.s., and on the other hand
$U_{i,j}$ can take only two values.  We shall therefore concentrate on
proving (\ref{secondo}) starting from the strong law of large numbers
for functionals of the process $\process$ (see \cite{Norris}, page
126), which in particular gives
\begin{equation}
\lim_{t\uparrow\infty}\,\frac{1}{t}\int_0^t |\eta_\tau|d\tau=\sum_{\eta\in
    S}|\eta|\pi(\eta)=\rho~~~~~~~\prob\mbox{-a.s.}.
\label{SLLN}
\end{equation}

To begin with, we notice that the particles still in the system after
the $j$th configuration jump are those that were present at the
beginning or that have been injected up to the $j$th change of
configuration and have not yet left the lattice. The following lemma
concerning the number of particles holds.
\begin{lemma}
\label{idTU}
$|\zeta_j|=\sum_{x\in\Lambda}\Theta_{1,j}(x)+\sum_{i=1}^j U_{i,j}$ for each $j\ge 1$.
\end{lemma}
\begin{proof}
For brevity, set
$\psi_i(x):=\mathds{1}(|\zeta_i|>|\zeta_{i-1}|)[1-\zeta_{i-1}(x)]\zeta_i(x)$
for each $i\ge 1$ and $x\in\Lambda$. For every $i\ge 1$ and
$y\in\Lambda$ we have
\begin{equation}
  \zeta_i(y)=\sum_{x\in\Lambda}\theta_i(x,y)+\psi_i(y).
  \label{eq11}
\end{equation}
This identity can be easily verified considering separately the
alternatives (I), (D), and (E) for the $i$th configuration jump. It
simply states that a particle in the system either was already present
before the last configuration jump or it has been injected during this
change of configuration.  Making use of (\ref{eq11}) we show by
induction that for all integers $j\ge 1$ and $i$ running from $j$ to 1
\begin{equation}
|\zeta_j|=\sum_{x\in\Lambda}\Theta_{i,j}(x)+U_{i,j}+\cdots+U_{j,j}.
\label{eq12}
\end{equation}
The lemma follows from this last formula when $i=1$.  In order to
demonstrate (\ref{eq12}), pick $j\ge 1$ and notice that
$\sum_{y\in\Lambda}\psi_j(y)=U_{j,j}$ since
$\sum_{y\in\Lambda}[1-\zeta_{j-1}(y)]\zeta_j(y)=1$ if
$|\zeta_j|>|\zeta_{j-1}|$, which corresponds to alternative (I).
Then, identity (\ref{eq11}) yields
\begin{equation*}
|\zeta_j|=\sum_{x\in\Lambda}\sum_{y\in\Lambda}\theta_j(x,y)+\sum_{y\in\Lambda}\psi_j(y)=\sum_{x\in\Lambda}\Theta_{j,j}(x)+U_{j,j}.
\end{equation*}
This proves (\ref{eq12}) when $i=j$.  Suppose now that (\ref{eq12})
holds with $i+1\le j$ in the place of $i$. From definition
(\ref{Thetadef}) we have that
$\sum_{y\in\Lambda}\theta_i(x,y)\,\Theta_{i+1,j}(y)=\Theta_{i,j}(x)$
and that $\Theta_{i+1,j}(y)$ is proportional to $\zeta_i(y)$, so that
in particular (\ref{Udef}) can be recast as
$\sum_{y\in\Lambda}\psi_i(y)\,\Theta_{i+1,j}(y)=U_{i,j}$. Then, we get
from the inductive hypothesis first and (\ref{eq11}) later that
\begin{eqnarray}
\nonumber
|\zeta_j|&=&\sum_{y\in\Lambda}\Theta_{i+1,j}(y)+U_{i+1,j}+\cdots+U_{j,j}\\
\nonumber
&=&\sum_{y\in\Lambda}\zeta_i(y)\,\Theta_{i+1,j}(y)+U_{i+1,j}+\cdots+U_{j,j}\\
\nonumber
&=&\sum_{x\in\Lambda}\sum_{y\in\Lambda}\theta_i(x,y)\,\Theta_{i+1,j}(y)+\sum_{y\in\Lambda}\psi_i(y)\,\Theta_{i+1,j}(y)+U_{i+1,j}+\cdots+U_{j,j}\\
\nonumber
&=&\sum_{x\in\Lambda}\Theta_{i,j}(x)+U_{i,j}+\cdots+U_{j,j}.
\end{eqnarray}
This proves (\ref{eq12}) when $i<j$. \qed
\end{proof}

The fact that $\eta_t=\zeta_j$ if $J_j\le t<J_{j+1}=J_j+H_j$ and that
$J_{N_t}\le t<J_{N_t+1}$ allows to show that
\begin{equation*}
  \int_0^t |\eta_\tau|d\tau=\sum_{j=0}^{N_t-1}H_j|\zeta_j|+\big(t-J_{N_t}\big)|\zeta_{N_t}|=
  \sum_{j=0}^{N_t}H_j|\zeta_j|+\big(t-J_{N_t+1}\big)|\zeta_{N_t}|.
\end{equation*}
Using for all $j\ge 1$ the identity
$|\zeta_j|=\sum_{x\in\Lambda}\Theta_{1,j}(x)+\sum_{i=1}^j U_{i,j}$
provided by Lemma \ref{idTU} we then find
\begin{equation*}
  \int_0^t |\eta_\tau|d\tau=H_0|\zeta_0|+\big(t-J_{N_t+1}\big)|\zeta_{N_t}|+\sum_{x\in\Lambda}\sum_{j=1}^{N_t}H_j\Theta_{1,j}(x)+
  \sum_{i=1}^{N_t}\sum_{j=i}^{N_t}H_jU_{i,j}.
\end{equation*}
This way, noticing that $0\le J_{N_t+1}-t\le
J_{N_t+1}-J_{N_t}=H_{N_t}$ we obtain the bound
\begin{eqnarray}
  \nonumber
  \Bigg|\sum_{i=1}^{N_t}\sum_{j=i}^\infty H_jU_{i,j}-\int_0^t |\eta_\tau|d\tau\Bigg|&\le&H_0|\zeta_0|+\big(J_{N_t+1}-t\big)|\zeta_{N_t}|+\\[-1.1em]
  \nonumber
  &&\sum_{x\in\Lambda}\sum_{j=1}^{N_t} H_j\Theta_{1,j}(x)+\sum_{i=1}^{N_t}\sum_{j=N_t+1}^\infty H_jU_{i,j}\\[0.3em]
\nonumber
  &\le&|\Lambda|H_0|+|\Lambda|H_{N_t}+\\
  &&\sum_{x\in\Lambda}\sum_{j=1}^{N_t} H_j\Theta_{1,j}(x)+\sum_{i=1}^{N_t}\sum_{j=N_t+1}^\infty H_jU_{i,j}.
\label{T_3}
\end{eqnarray}
The limit (\ref{secondo}) follows from (\ref{SLLN}) if we prove that
the r.h.s.\ of (\ref{T_3}) divided by $t$ goes to zero
$\prob$-a.s.\ when $t$ is sent to infinity. It is clear that
$\lim_{t\uparrow\infty}(1/t)H_0=0$ $\prob$-a.s.\ since $H_0<\infty$
$\prob$-a.s.. Then, we must show that
$\lim_{t\uparrow\infty}(1/t)V_{N_t}=0$ $\prob$-a.s. \linebreak with $V_n$ once
equal to $H_n$, once equal to $\sum_{x\in\Lambda}\sum_{j=1}^n
H_j\Theta_{1,j}(x)$, and once equal to
$\sum_{i=1}^n\sum_{j=n+1}^\infty H_jU_{i,j}$. The average number of
configuration jumps per unit time is
$\lim_{t\uparrow\infty}(1/t)N_t=\sum_{\eta\in S}q(\eta)\pi(\eta)$
$\prob$-a.s.\ with $q(\eta)$ as in (\ref{qdef}) (see \cite{Serfozo},
page 265). As $\sum_{\eta\in S}q(\eta)\pi(\eta)<\infty$, we obtain
$\lim_{t\uparrow\infty}(1/t)V_{N_t}=0$ $\prob$-a.s.\ if we demonstrate
that $\lim_{n\uparrow\infty}(1/n)V_n=0$ $\prob$-a.s..  The
Borel-Cantelli lemma states that $\lim_{n\uparrow\infty}(1/n)V_n=0$
$\prob$-a.s.\ if $\sum_{n=1}^\infty\prob[V_n>\epsilon n]<\infty$ for
all $\epsilon>0$ and the Markov's inequality yields
$\prob[V_n>\epsilon n]\le(1/\epsilon n)^2\,\Ex[V_n^2]$ for every $n\ge
1$ and $\epsilon>0$. This way, we conclude that
$\lim_{t\uparrow\infty}(1/t)V_{N_t}=0$ $\prob$-a.s.\ if there exists a
positive constant $C<\infty$ such that $\Ex[V_n^2]\le C$ for all $n\ge
1$.  Let us show that such a constant exists.  We recall that
$H_0,\ldots,H_i$ are independent exponential random variables of
parameters $q(\zeta_0),\ldots,q(\zeta_i)$ conditional on
$\zeta_0,\ldots,\zeta_i$. We set $\delta:=\min_{\eta\in S}\{q(\eta)\}$
and we observe that $\delta>0$ since $q(\eta)>0$ for all $\eta$
belonging to the finite set $S$.

We have that for all $n\ge 1$
\begin{eqnarray}
  \nonumber
  \Ex\big[H_n^2\big]&=&\sum_{\eta_0\in S}\cdots\sum_{\eta_n\in S}\Ex\big[H_n^2\big|\zeta_0=\eta_0\land\ldots\land\zeta_n=\eta_n\big]\,
  \prob\big[\zeta_0=\eta_0\land\ldots\land\zeta_n=\eta_n\big]\\
  \nonumber
  &=&\sum_{\eta_0\in S}\cdots\sum_{\eta_n\in S}\frac{2}{q^2(\eta_n)}\,\prob\big[\zeta_0=\eta_0\land\ldots\land\zeta_n=\eta_n\big]
  \le\frac{2}{\delta^2}.
\end{eqnarray}
Thus, there exists $C<\infty$ such that $\Ex[V_n^2]\le C$ for each
$n\ge 1$ when $V_n:=H_n$. The cases
$V_n:=\sum_{x\in\Lambda}\sum_{j=1}^n H_j\Theta_{1,j}(x)$ and
$V_n:=\sum_{i=1}^n\sum_{j=n+1}^\infty H_jU_{i,j}$ are more involved
and require the use of the following lemma.
\begin{lemma}
\label{U_lem}
There exist positive constants $c<\infty$ and $r<1$ with the property
that $\Ex[\Theta_{i,j}(x)]\le c\,r^{j-i}$ for all $j\ge i\ge 1$ and
$x\in\Lambda$.
\end{lemma}
\begin{proof}
Let $E$ be the vector space of the functions $f:\Lambda\times
S\to\mathbb{C}$ endowed with the norm
$\|f\|:=\max_{(x,\eta)\in\Lambda\times S}\{|f(x,\eta)|\}$ and let
$W:E\to E$ be the linear operator defined for each $f\in E$ by
\begin{eqnarray}
\nonumber
Wf(x,\eta)&:=&\sum_{\eta'\in S}\eta(x)\eta'(x)\prob\big[\zeta_1=\eta'\big|\zeta_0=\eta\big]f(x,\eta')+\\
\nonumber
&&\sum_{y\in\Lambda}\sum_{\eta'\in S}\eta(x)\big[1-\eta(y)\big]\big[1-\eta'(x)\big]\eta'(y)
\prob\big[\zeta_1=\eta'\big|\zeta_0=\eta\big]f(y,\eta').
\end{eqnarray}
Denoting by $\|W\|:=\sup_{f\in E\,:\,\|f\|=1}\{\|Wf\|\}$ the norm of
$W$ induced by the vector norm on $E$ and by $\sigma(W)$ the spectrum
of $W$, Gelfand's formula for the spectral radius states that
$\lim_{n\uparrow\infty}\|W^n\|^{1/n}=\max_{\xi\in\sigma(W)}\{|\xi|\}$.
The powers
of $W$ are related to certain expected values, as we shall see in a
moment.

Pick a function $f\in E$ and consider the random variable
$\Theta^f_{i,j}(x_{i-1})$ defined for each couple of integers $j\ge
i\ge 1$ and site $x_{i-1}\in\Lambda$ by
\begin{equation*}
\Theta^f_{i,j}(x_{i-1}):=\sum_{x_i\in\Lambda}\cdots\sum_{x_j\in\Lambda}\prod_{k=i}^j\theta_k(x_{k-1},x_k)f(x_j,\zeta_j).
\end{equation*}
The variable $\Theta^f_{i,j}(x_{i-1})$ reduces to
$\Theta_{i,j}(x_{i-1})$ given by (\ref{Thetadef}) when $f$ is
identically equal to one. The Markov property of $\jump$ in
combination with the fact that $\theta_k(x,y)$ is a deterministic
function of only $\zeta_{k-1}$ and $\zeta_k$ yields for each
$x_{i-1}\in\Lambda$ and $\eta_{i-1}\in S$ the relationship
\begin{eqnarray}
  \nonumber
  &&\Ex\big[\Theta^f_{i,j}(x_{i-1})\big|\zeta_{i-1}=\eta_{i-1}\big]=\\
\nonumber
  &=&\sum_{x_i\in\Lambda}\cdots\sum_{x_j\in\Lambda}\sum_{\eta_i\in S}\cdots\sum_{\eta_j\in S}
  \prod_{k=i}^j\Ex\Big[\theta_k(x_{k-1},x_k)\mathds{1}(\zeta_k=\eta_k)\Big|\zeta_{k-1}=\eta_{k-1}\Big]f(x_j,\eta_j).
\end{eqnarray}
It follows from here that
$\Ex[\Theta^f_{i,j}(x_{i-1})|\zeta_{i-1}=\eta_{i-1}]=W^{j-i+1}f(x_{i-1},\eta_{i-1})$,
with $W$ the above linear operator, since the homogeneity of the jump
chain implies that for each $k\ge 1$ and function $f$
\begin{eqnarray}
\nonumber
&&\sum_{y\in\Lambda}\sum_{\eta'\in S}\Ex\Big[\theta_k(x,y)\mathds{1}(\zeta_k=\eta')\Big|\zeta_{k-1}=\eta\Big]f(y,\eta')=\\
\nonumber
&=&\sum_{\eta'\in S}\eta(x)\eta'(x)\prob\big[\zeta_k=\eta'\big|\zeta_{k-1}=\eta\big]f(x,\eta')+\\
\nonumber
&&\sum_{y\in\Lambda}\sum_{\eta'\in S}\eta(x)\big[1-\eta(y)\big]\big[1-\eta'(x)\big]\eta'(y)
\prob\big[\zeta_k=\eta'\big|\zeta_{k-1}=\eta\big]f(y,\eta')\\
\nonumber
&=&Wf(x,\eta).
\end{eqnarray}
In conclusion, for every $j\ge i\ge 1$ and $x\in\Lambda$ we find
\begin{equation}
\Ex\big[\Theta^f_{i,j}(x)\big]=\sum_{\eta\in S}\prob\big[\zeta_{i-1}=\eta\big]\,W^{j-i+1}f(x,\eta).
\label{ExThetaf}
\end{equation}

We now prove that $|\xi|<1$ for each eigenvalue $\xi\in\sigma(W)$. To
this aim, we recall that there exist $\prob$-a.s.\ infinitely many $k$
such that $\theta_k(x,y)=0$ for all $x$ and $y$ due to recurrence of
$\jump$. Consequently, $\lim_{j\uparrow\infty}\Theta^f_{i,j}(x)=0$
$\prob$-a.s.. In its turn, Lebesgue's dominated convergence theorem
gives $\lim_{j\uparrow\infty}\Ex[\Theta^f_{i,j}(x)]=0$ since
$|\Theta^f_{i,j}(x)|\le\|f\|$. This is true for every $i\ge 1$,
$x\in\Lambda$, and $f\in E$ irrespective of the distribution of
$\zeta_0=\eta_0$. Pick $\xi\in\sigma(W)$ and let $f\in E$ be a
corresponding eigenfunction, so that $Wf=\xi f$. There exists a pair
$(\bar{x},\bar{\eta})\in\Lambda\times S$ such that
$f(\bar{x},\bar{\eta})\ne 0$ and we can assume without loss of
generality that $f(\bar{x},\bar{\eta})=1$.  For the process $\process$
defined by the initial condition $\eta_0=\bar{\eta}$ $\prob$-a.s.\ the
application of (\ref{ExThetaf}) with $i=1$ and $f$ the previously
introduced eigenfunction yields
$\Ex\big[\Theta^f_{1,j}(\bar{x})\big]=\xi^j$.  This way, the bound
$|\xi|<1$ follows from
$\lim_{j\uparrow\infty}\Ex[\Theta^f_{1,j}(\bar{x})]=0$.

The fact that $\max_{\xi\in\sigma(W)}\{|\xi|\}<1$ in combination with
Gelfand's spectral radius formula proves that there exist positive
constants $c<\infty$ and $r<1$ such that $\|W^{n+1}\|\le c\,r^n$ for
all $n\ge 0$. Expression (\ref{ExThetaf}) with $f$ identically equal
to one, which has norm $\|f\|=1$, shows that for all $j\ge i\ge 1$ and
$x\in\Lambda$
\begin{equation*}
\Ex\big[\Theta_{i,j}(x)\big]\le\sum_{\eta\in S}\prob\big[\zeta_{i-1}=\eta\big]\big\|W^{j-i+1}\big\|\le c\,r^{j-i}
\end{equation*}
This concludes the proof. \qed
\end{proof}

Let us now set $V_n:=\sum_{x\in\Lambda}\sum_{j=1}^n
H_j\Theta_{1,j}(x)$ for all $n\ge 1$. Since for any couple of integers
$i\ge 1$ and $j\ge 1$ the binary random variables $\Theta_{1,i}(x)$
and $\Theta_{1,j}(y)$ are deterministic functions of
$\zeta_0,\ldots,\zeta_k$ with $k:=\max\{i,j\}$, we have the bound
\begin{eqnarray}
  \nonumber
  &&\Ex\Big[H_iH_j\Theta_{1,i}(x)\Theta_{1,j}(y)\Big]=\\
\nonumber
  &=&\sum_{\eta_0\in S}\cdots\sum_{\eta_k\in S}\Ex\bigg[H_iH_j\Theta_{1,i}(x)\Theta_{1,j}(y)\prod_{l=0}^k\mathds{1}(\zeta_l=\eta_l)\bigg]\\
  \nonumber
  &=&\sum_{\eta_0\in S}\cdots\sum_{\eta_k\in S}
  \Ex\Big[H_iH_j\Big|\zeta_0=\eta_0\land\ldots\land\zeta_k=\eta_k\Big]
  \Ex\bigg[\Theta_{1,i}(x)\Theta_{1,j}(y)\prod_{l=0}^k\mathds{1}(\zeta_l=\eta_l)\bigg]\\
  \nonumber
  &\le&\sum_{\eta_0\in S}\cdots\sum_{\eta_k\in S}\frac{2}{q(\eta_i)q(\eta_j)}\,
\Ex\bigg[\Theta_{1,i}(x)\Theta_{1,j}(y)\prod_{l=0}^k\mathds{1}(\zeta_l=\eta_l)\bigg]\\
  &\le&\frac{2}{\delta^2}\,\Ex\big[\Theta_{1,i}(x)\Theta_{1,j}(y)\big]
  \le\frac{2}{\delta^2}\sqrt{\Ex\big[\Theta_{1,i}(x)\big]\Ex\big[\Theta_{1,j}(y)\big]},
  \label{T_4}
\end{eqnarray}
where the Cauchy-Schwarz inequality has been exploited to obtain the
last inequality. Thus, combining (\ref{T_4}) with lemma \ref{U_lem} we
arrive at the result
\begin{eqnarray}
  \nonumber
  \Ex\Bigg[\bigg(\sum_{x\in\Lambda}\sum_{j=1}^n H_j\Theta_{1,j}(x)\bigg)^2\Bigg]&=&
  \sum_{x\in\Lambda}\sum_{i=1}^n\sum_{y\in\Lambda}\sum_{j=1}^n\Ex\Big[H_iH_j\Theta_{1,i}(x)\Theta_{1,j}(y)\Big]\\
  \nonumber
  &\le&\frac{2}{\delta^2}\Bigg(\sum_{x\in\Lambda}\sum_{j=1}^n\sqrt{\Ex\big[\Theta_{1,j}(x)\big]}\Bigg)^2
  \le\frac{2c|\Lambda|^2}{\delta^2(1-\sqrt{r})^2}.
\end{eqnarray}
This proves that there exists $C<\infty$ such that $\Ex[V_n^2]\le C$
for each $n\ge 1$.

To conclude, set $V_n:=\sum_{i=1}^n\sum_{j=n+1}^\infty H_jU_{i,j}$ for
all $n\ge 1$. The same arguments that have led to (\ref{T_4}) show
that $\delta^2\Ex[H_jH_kU_{i,j}U_{h,k}]\le 2$ for all $j\ge i\ge 1$
and $k\ge h\ge 1$ since the binary random variables $U_{i,j}$ and
$U_{h,k}$ are deterministic functions of
$\zeta_{\min\{i,h\}-1},\ldots,\zeta_{\max\{j,k\}}$. Consequently, we
can write
\begin{eqnarray}
  \nonumber
  \Ex\Bigg[\bigg(\sum_{i=1}^n\sum_{j=n+1}^\infty H_jU_{i,j}\bigg)^2\Bigg]&=&
  \sum_{i=1}^n\sum_{j=n+1}^\infty\sum_{h=1}^n\sum_{k=n+1}^\infty\Ex\Big[H_jH_kU_{i,j}U_{h,k}\Big]\\
  \nonumber
  &\le&\frac{2}{\delta^2}\sum_{i=1}^n\sum_{j=n+1}^\infty\sum_{h=1}^n\sum_{k=n+1}^\infty\Ex\big[U_{i,j}U_{h,k}\big]\\
  \nonumber
  &\le&\frac{2}{\delta^2}\sum_{i=1}^n\sum_{j=n+1}^\infty\sum_{h=1}^n\sum_{k=n+1}^\infty\sqrt{\Ex[U_{i,j}]\,\Ex[U_{h,k}]}\\
  &=&\frac{2}{\delta^2}\Bigg(\sum_{i=1}^n\sum_{j=n+1}^\infty\sqrt{\Ex[U_{i,j}]}\Bigg)^2.
\label{T_5}
\end{eqnarray}
On the other hand, from (\ref{Udef}) we have
$U_{i,j}\le\sum_{x\in\Lambda}\Theta_{i+1,j}(x)$ for all $j>i$, giving
$\Ex[U_{i,j}]\le|\Lambda|c\,r^{j-i-1}$ thanks to lemma \ref{U_lem}.
Combining this bound with (\ref{T_5}) we get
\begin{equation*}
  \Ex\Bigg[\bigg(\sum_{i=1}^n\sum_{j=n+1}^\infty H_jU_{i,j}\bigg)^2\Bigg]\le\frac{2c|\Lambda|}{\delta^2(1-\sqrt{r})^4}.
\end{equation*}
Thus, there exists $C<\infty$ with the property that $\Ex[V_n^2]\le C$
for any $n\ge 1$.

\end{document}